%% file: conference_101719.tex
\documentclass[conference,letterpaper]{IEEEtran}
\IEEEoverridecommandlockouts

\usepackage{cite}
\usepackage{amsmath,amssymb,amsfonts}
\usepackage{algorithmic}
\usepackage{graphicx}
\usepackage{textcomp}
\usepackage{xcolor}
\def\BibTeX{{\rm B\kern-.05em{\sc i\kern-.025em b}\kern-.08em
    T\kern-.1667em\lower.7ex\hbox{E}\kern-.125emX}}

\usepackage[font=footnotesize]{caption}

\usepackage{amsfonts}
\usepackage{amssymb}
\usepackage{algorithm}
\usepackage{textcomp}
\usepackage{verbatim}

\usepackage{amsthm}
\usepackage{pgfplots}
\pgfplotsset{compat=1.17}
\usepgfplotslibrary{groupplots}
\usepackage{hyperref}
\hypersetup{
  colorlinks=false,
  linkbordercolor=white,
 urlbordercolor=white,
pdfborder={0 0 0}
}
\usetikzlibrary{patterns}
\usetikzlibrary{arrows}
\usetikzlibrary{calc} 
\usetikzlibrary{positioning}
\usetikzlibrary{arrows.meta}

\usepackage{subcaption}
\usepackage{tikz-3dplot}

\usetikzlibrary{pgfplots.groupplots}
\usepackage{xstring}

\newtheorem{corollary}{Corollary}[section]
\newtheorem{lemma}{Lemma}[section]
\newtheorem{proposition}{Proposition}[section]

\pgfdeclarepatternformonly[\LineSpace]{my north east lines}{\pgfqpoint{-1pt}{-1pt}}{\pgfqpoint{\LineSpace}{\LineSpace}}{\pgfqpoint{\LineSpace}{\LineSpace}}%
{
    \pgfsetlinewidth{0.4pt}
    \pgfpathmoveto{\pgfqpoint{0pt}{0pt}}
    \pgfpathlineto{\pgfqpoint{\LineSpace + 0.1pt}{\LineSpace + 0.1pt}}
    \pgfusepath{stroke}
}

\newdimen\LineSpace
\tikzset{
    line space/.code={\LineSpace=#1},
    line space=3pt
}

\tikzset{cross/.style={cross out, draw, 
         minimum size=2*(#1-\pgflinewidth), 
         inner sep=0pt, outer sep=0pt}}

\begin{document}

\title{Optimal Antenna Placement for Two-Antenna
Near-Field Wireless Power Transfer
\thanks{The work was supported in part by the German Research Foundation through Project SFB 1483.}
}

\author{
\IEEEauthorblockN{Kenneth MacSporran Mayer, Laura Cottatellucci, and Robert Schober}
\IEEEauthorblockA{Friedrich-Alexander University Erlangen-Nuremberg, Germany}
}

\maketitle

\begin{abstract}
\input{Text/Abstract}
\end{abstract}

\input{Text/Introduction}

\input{Text/System_Model}

\input{Text/Problem_Formulation_and_Optimal_Solution}

\input{Text/Results_and_Performance_Evalutation}

\input{Text/Conclusion}
\bibliographystyle{ieeetr}
\bibliography{IEEEabrv,refs}

\end{document}

%% file: Text/Abstract.tex
Current trends in communication system design precipitate a change in the operating regime from the traditional far-field to the radiating near-field (Fresnel) region. We investigate the optimal transmit antenna placement for a multiple-input single-output (MISO) wireless power transfer (WPT) system designed for a three-dimensional cuboid room under line-of-sight (LoS) conditions in the Fresnel region. We formulate an optimisation problem for maximising the received power at the worst possible receiver location by considering the spherical nature of the electromagnetic (EM) wavefronts in the Fresnel region while assuming perfect knowledge of the channel at the transmitter. For the case of two transmit antennas, we derive a closed-form expression for the optimal positioning of the antennas which is purely determined by the geometry of the environment. If the room contains locations where the far-field approximation holds, the proposed positioning is shown to reduce to the far-field solution. The analytical solution is validated through simulation. Furthermore, the maximum received power at the locations yielding the worst performance is quantified and the power gain over the optimal far-field solution is presented. For the considered cuboid environment, we show that a distributed antenna system is optimal in the Fresnel region, whereas a co-located antenna architecture is ideal for the far-field.  

%% file: Text/Introduction.tex
\section{Introduction}\label{Section: Introduction}
Future communication systems are moving towards operating at higher frequencies, as the associated large bandwidth helps support the ever-increasing requirements on data rate, low latency, network heterogeneity, as well as energy efficiency. 
This precipitates a change in the operating regime from the traditional far-field to the radiating near-field (Fresnel) region, which must be reflected in the modelling of the wireless channel \cite{Zhang2022a}.
Wireless systems designed for the Fresnel region are of interest not only for communications \cite{Zhang2021}, but also for wireless power transfer (WPT) \cite{Zhanga}.
WPT systems designed for the Fresnel region are capable of focusing the energy beams for transferring power wirelessly, thus yielding a larger amount of received power and causing less energy pollution in unwanted directions compared to traditional far-field approaches \cite{Zhang2022a}, \cite{Zhanga}. 

A key difference between the far-field and the radiating near-field is how the electromagnetic (EM) wavefronts are modelled. While a planar wavefront model (PWM) is suitable for the far-field, considering the spherical nature of the EM wavefronts is indispensable in the Fresnel region \cite{Zhang2022a}. 
Therefore, the design of wireless systems for the radiating near-field is driven by the spherical wavefront model (SWM) \cite{Zhang2022a, Zhang2021, Zhanga}.
The SWM has been considered for WPT in \cite{Zhang} where a dynamic metasurface antenna is employed to maximise the weighted sum of received energies via beam focusing.

While the study of WPT in the Fresnel region is at an early stage, far-field WPT has been extensively studied. A comprehensive overview of the design concepts, capabilities and limitations, prototypes, and applications of WPT systems is provided in \cite{Massa2013,6810996, Costanzo2016SmartSI,Clerckx, Clerckx2018, Clerckx2022, Gu2021}. 
Typically, WPT systems are required to have line-of-sight (LoS) between transmitter and receiver in order to attain adequate power transfer efficiency \cite{Clerckx2022}.
Additionally, in state-of-the-art WPT systems, the transmit antennas of the energy transmitter are either co-located or distributed \cite{Clerckx2022}.
Co-located antenna architectures have been employed, for example, in \cite{6954434} and distributed antenna systems (DASs) have been considered in \cite{ Choi2018, Shen2021}.
DASs are an appealing transmit antenna architecture for WPT systems and enable cooperation among the distributed transmit antennas. Compared to co-located energy transmitters, employing a DAS results in a more uniform distribution of the received power in the environment \cite{Clerckx2022}. DASs for WPT have been investigated for different transmission strategies such as transmit antenna selection \cite{Shen2021} and maximum ratio transmission \cite{Choi2018}. 

In this paper, we consider a multiple-input single-output (MISO) WPT system that comprises a two-antenna energy transmitter and a single antenna energy receiver in LoS conditions. The energy receiver is located at an arbitrary position in a three-dimensional cuboid room and the transmit antennas are located on one of the room's walls.
We determine the amplitude variations of the wireless channel based on geometrical considerations and optimise the positions of the transmit antennas analytically, thereby showing whether a co-located or a distributed transmit antenna architecture is optimal. 
Hereby, the objective is to maximise the received power for the worst possible receiver position. 
Thus, the proposed solution ensures provision of the maximum, worst-case power for the given environment
when there is no prior knowledge on the receiver's location in the room. This is crucial for wirelessly powered devices utilised for continuous monitoring, such as sensors and wearables. 

Our approach for optimising the transmit antenna positions of a WPT system is designed for a three-dimensional cuboid room and relies on exploiting the spherical nature of the EM wavefronts.
Therefore, the existing results \cite{Massa2013,6810996, Costanzo2016SmartSI,Clerckx, Clerckx2018, Clerckx2022, Gu2021, 6954434, Choi2018, Shen2021} which were obtained for WPT systems designed for the far-field operating regime are not applicable to the problem considered in this paper. 
The antenna placement problem is formulated for a general number of transmit antennas and we solve the problem optimally for a two-antenna system. 
The analytical expression describing the optimal transmit antenna positions only depends on the geometry of the environment. 
While the proposed solution is designed around capturing the effects of the spherical EM wavefronts in the Fresnel region, the solution is shown to converge to the optimal far-field solution once the distances are large enough for the PWM to become sufficiently accurate at the worst possible receiver positions. 
The proposed optimal transmit antenna deployment reveals that a DAS is the optimal transmit antenna architecture in the Fresnel region, whereas a co-located architecture is optimal in the far-field. 
The point of transition between the DAS and the co-located transmit antenna architecture is determined through geometrical considerations.
Moreover, the extension of the presented methodology to systems with more than two transmit antennas is briefly discussed.
The analytical solution is validated by solving the problem numerically. Furthermore, the power gain over the far-field solution is investigated which shows that, in the Fresnel region, deploying a DAS results in up to three times as much received energy at the receiver compared to a co-located architecture.

%% file: Text/System_Model.tex
\section{System Model}\label{Section: System Model}
The WPT system considered in this paper comprises a transmitter equipped with $N_t$ antennas and a single-antenna receiver endowed with an energy harvester. 
An illustration of the environment is provided in Fig. \ref{fig: System Model Overview} for $N_t=2$. 
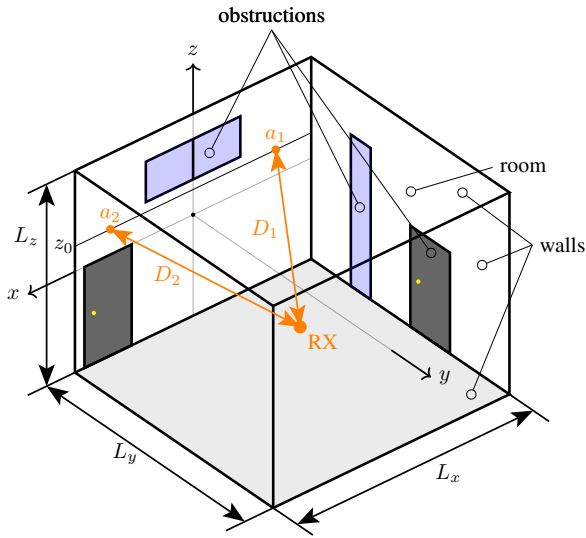
\begin{figure}
    \centering
    \scalebox{0.82}{
    \input{Figures/system_model_3d.tex}
    }
    \caption{ Illustration of the system model with two transmit antennas, i.e., $N_t=2$. The receiver (RX) is located in a cuboid room, where portions of the wall are obstructed. The transmit antennas are located on a horizontal line at $a_1$ and $a_2$. The distance to the receiver is given by $D_1$ and $D_2$, respectively. }
    \label{fig: System Model Overview}
\end{figure}

\subsection{Environment}\label{Subsection: Environment Model}
The considered environment is a three-dimensional cuboid room which is defined by the Cartesian product of the sets $\mathcal{X}$, $\mathcal{Y}$, and $\mathcal{Z}$, i.e., $\mathcal{X} \times \mathcal{Y} \times \mathcal{Z}$, with
\begin{align}\label{eq: Receiver position in X}
    \mathcal{X} = \left\{ x \, \Bigg\vert \, -\frac{L_x}{2} \leq x \leq \frac{L_x}{2} \; \text{with} \; L_x > 0 \right\},
\end{align}
\begin{align}\label{eq: Receiver position in Y}
    \mathcal{Y} = \left\{ y \, \big\vert \, 0 \leq y \leq L_y \; \text{with} \; L_y > 0 \right\},
\end{align}
and
\begin{align}\label{eq: Receiver position in Z}
    \mathcal{Z} = \left\{ z \, \Bigg\vert \, -\frac{L_z}{2} \leq z \leq \frac{L_z}{2} \; \text{with} \; L_z > 0 \right\}.
\end{align}
The receiver, e.g., a wearable device, relies on being powered wirelessly and is assumed to be located at an arbitrary position in the environment. The receiver location is defined by the triplet $(x,y,z) \in \mathcal{X} \times \mathcal{Y} \times \mathcal{Z}$. 

\subsection{Transmit Antenna Model}\label{Subsection: Transmit Antenna Model}
Identical and omnidirectioal transmit antennas are employed to illuminate the entire room since the receiver may be located anywhere in the environment.
The location of transmit antenna $i$ is described by the triplet $(a_{ix},a_{iy},a_{iz}) \in \mathcal{X} \times \mathcal{Y} \times \mathcal{Z}$, $\forall i=1,\dots,N_t$. 
The amount of power obtained at the receiver is shown to depend on the placement of the transmit antennas in Subsection \ref{Subsection: Channel Model}. Consequently, the amount of power at the receiver can be maximised by optimising the positions of the transmit antennas. 

There is an inherent trade-off between the achievable performance and the practicality of the transmit antenna deployment. 
Although leveraging the full potential of a flexible placement would yield the best performance, the resulting transmit antenna locations may be impractical or even infeasible in practice. 
Therefore, we restrict the antennas to being placed along a horizontal line defined by $y=0$, $z=z_0 \in \mathcal{Z}$. Consequently, the transmit antenna positions are restricted in the environment through the following condition
\begin{align}\label{eq: TX antenna positions in X}
    &\forall i=1,\dots,N_t: a_{ix} \in \mathcal{X}, a_{iy} = 0,  a_{iz} = z_0 \in \mathcal{Z}.
\end{align}
Throughout this paper, $a_{ix}$ is denoted by $a_{i}$ for simplicity of notation. 

\subsection{Channel and Signal Model}\label{Subsection: Channel Model}
The Fresnel region of a wireless system is defined, e.g., in \cite{Zhang2021}, and depends on the relationship between the wavelength of the EM waves and the distance between the transmit antennas in comparison to their distance from the receiver. 
The size of the Fresnel region increases with the carrier frequency. While the Fresnel region is negligibly small in conventional wireless systems, when using mmWave frequency bands or higher, the size of the Fresnel region becomes relevant, e.g., for indoor scenarios. 
Furthermore, for high carrier frequencies, the severe reflection and scattering losses cause the wireless channel to become predominantly LoS \cite{Zhang2022a,Zhanga, Do21}. A LoS channel is typically assumed for wireless systems operating in the Fresnel region \cite{Li2022, Zhanga, Zhang2021}. 

LoS between transmitter and receiver allows a WPT system to attain adequate power transfer efficiency \cite{Clerckx2022}. 
On the other hand, we note that LoS channels are susceptible to blockages \cite{Zhanga}. For analytical tractability of system design, we make the assumption that a LoS connection between the transmit antennas and the receiver exists and do not consider the impact of channel blockages in this paper.

Based on the previous considerations, free-space LoS propagation of the EM wavefronts is considered in this paper. 
When considering free-space LoS conditions, the equivalent complex baseband channel between transmit antenna $i$ and the receiver is given by \cite{Tse2005a}
\begin{align}\label{eq: LOS Channel}
    g_i = \frac{\sqrt{c}}{D_i} \mathrm{e}^{-\mathrm{j}\frac{2\pi}{\lambda} D_i},
\end{align}
where $\lambda$ is the wavelength, $D_i$ is the distance between transmit antenna $i$ and the receiver, and $c$ is the channel power gain at the reference distance of 1 meter. 
Since the wireless channel of the considered MISO system is modelled using the SWM, both the amplitude and the phase of $g_i$ depend on distance $D_i$. In contrast, for the PWM, the \emph{path loss} between all transmit antennas and the receiver is assumed to be constant, i.e., $D_i \approx D, \forall i=1\dots N_t$.
The MISO channel is represented by the $N_t$-dimensional vector $\boldsymbol{g} = [g_1,\dots,g_{N_t}]^T \in \mathbb{C}^{{N_t} \times 1}$.
The channel model in \eqref{eq: LOS Channel} has been considered for describing systems designed for the Fresnel region in \cite{Li2022} and was adapted for the case of directional antennas in \cite{Zhanga, Zhang2021}. 
In \cite{Zhang2013}, it was shown that the received power $\gamma$ for any given position in the environment is maximised by transmitting vector $\boldsymbol{s} = \sqrt{P}\boldsymbol{g}/\lVert \boldsymbol{g} \rVert_2 \in \mathbb{C}^{N_t \times 1}$ over the $N_t$ transmit antennas, where $P$ is the total transmit power, which is shared among the antennas, and $\lVert \cdot \rVert_2$ is the Euclidean norm. Therefore, the received power $\gamma$ is given by
\begin{align}\label{eq: RX Signal Power}
    \gamma = \lVert \boldsymbol{g}^H \boldsymbol{s} \rVert_2^2 = P \,\lVert \boldsymbol{g} \rVert_2^2 = Pc\sum_{i=1}^{N_t} \frac{1}{D_i^2}.
\end{align}
Thus, the received power is directly proportional to the sum of the squared, inverse distances between the receiver and the transmit antennas. 
While the method in \cite{Zhang2013} was designed for steering beams in the far-field, it also applies to focusing energy beams towards the receiver in the Fresnel region \cite{Zhang2021}. 
Note that the impact of additive noise on the received power is negligible.

%% file: Figures/system_model_3d.tex

\definecolor{FAU_BLAU}{RGB}{0,32,96}
\tdplotsetmaincoords{55}{140}
\begin{tikzpicture}[
tdplot_main_coords,
cube/.style={very thick,black},
cube_continued/.style={thick,black},
indicator/.style={<->,thick,arrows = {Stealth[length=10pt, inset=2pt]-Stealth[length=10pt, inset=2pt]},thick,black},
grid/.style={very thin,gray},
axis_outside/.style={->,black,thick},
axis_inside/.style={very thin,opacity=0.33}
]
\draw[axis_inside] (0,0,2) -- (5,0,2) ;
\draw[axis_outside] (5,0,2) -- (6,0,2) node[anchor=east]{$x$};
\draw[axis_inside] (2.5,0,2) -- (2.5,5,2) ;
\draw[axis_outside] (2.5,5,2) -- (2.5,6,2) node[anchor=west]{$y$};
\draw[axis_inside] (2.5,0,0) -- (2.5,0,4) ;
\draw[axis_outside] (2.5,0,4) -- (2.5,0,5) node[anchor=south]{$z$};

\draw[indicator] (5.6,0,0) -- (5.6,5,0) node[anchor=east,midway]{$L_y$} ;
\draw[cube_continued] (5,0,0) -- (6,0,0) ;
\draw[cube_continued] (5,5,0) -- (6,5,0) ;

\draw[indicator] (0,5.6,0) -- (5,5.6,0) node[anchor=west,midway,xshift=2mm,yshift=-1mm]{$L_x$} ;
\draw[cube_continued] (5,5,0) -- (5,6,0) ;
\draw[cube_continued] (0,5,0) -- (0,6,0) ;

\draw[indicator] (5.6,0,0) -- (5.6,0,4) node[anchor=east,midway,yshift=8mm]{$L_z$} ;
\draw[cube_continued] (5,0,4) -- (6,0,4) ;

\draw[cube] (0,0,0) -- (0,5,0) -- (5,5,0) -- (5,0,0) -- cycle;
\fill[brown,cube,opacity=0.08] (0,0,0) -- (0,5,0) -- (5,5,0) -- (5,0,0) -- cycle;
\draw[cube] (0,0,4) -- (0,5,4) -- (5,5,4) -- (5,0,4) -- cycle;

\draw[cube] (0,0,0) -- (0,0,4);
\draw[cube] (0,5,0) -- (0,5,4);
\draw[cube] (5,0,0) -- (5,0,4);
\draw[cube] (5,5,0) -- (5,5,4);

\draw[cube] (0,2.5,0) -- (0,3.5,0) -- (0,3.5,2) -- (0,2.5,2) -- cycle;
\fill[brown,cube,opacity=0.6] (0,2.5,0) -- (0,3.5,0) -- (0,3.5,2) -- (0,2.5,2) -- cycle;
\fill[yellow] (0,2.7,1) circle (1pt);

\draw[cube] (3.8,0,0) -- (4.8,0,0) -- (4.8,0,2) -- (3.8,0,2) -- cycle;
\fill[brown,cube,opacity=0.6] (3.8,0,0) -- (4.8,0,0) -- (4.8,0,2) -- (3.8,0,2) -- cycle;
\fill[yellow] (4.6,0,1) circle (1pt);

\draw[cube] (1.5,0,2.7) -- (1.5,0,3.5) -- (3.5,0,3.5) -- (3.5,0,2.7) -- cycle;
\fill[blue,opacity=0.2] (1.5,0,2.7) -- (1.5,0,3.5) -- (3.5,0,3.5) -- (3.5,0,2.7) -- cycle;
\draw[cube] (2.5,0,2.7) -- (2.5,0,3.5);

\draw[cube] (0,1,0) -- (0,1.5,0) -- (0,1.5,3) -- (0,1,3) -- cycle;
\fill[blue,opacity=0.2] (0,1,0) -- (0,1.5,0) -- (0,1.5,3) -- (0,1,3) -- cycle;

\draw[black] (0,0,2.5) -- (5,0,2.5) node[left,xshift=1mm] {$z_0$};

\coordinate (wall_label) at (-3,2,0);
\coordinate (room_label) at (-3,1,1);
\coordinate (obstructions) at (0,-1,4);
\coordinate (origin) at (2.5,0,2);
\coordinate (RX) at (2.75,3,1.5);
\coordinate (a_1) at (0.75,0,2.5);
\coordinate (a_2) at (4.25,0,2.5);
 
   

    

\draw[orange,<->,thick,arrows = {Stealth[length=10pt, inset=2pt]-Stealth[length=10pt, inset=2pt]}] (a_1) -- (RX) node[pos=0.45,left] {$D_1$};
\draw[orange,<->,thick,arrows = {Stealth[length=10pt, inset=2pt]-Stealth[length=10pt, inset=2pt]}] (a_2) -- (RX) node[pos=0.3,below] {$D_2$};

\fill (origin) circle (1pt);
\fill[orange] (RX) circle (3pt) node[below right] {RX};
\fill[orange] (a_1) circle (2pt) node[above ] {$a_1$};
\fill[orange] (a_2) circle (2pt) node[above ] {$a_2$};

\draw[o-] (-1,3,1) -- (wall_label) node[right] {walls};
\draw[o-] (1.5,5.5,5) -- (wall_label) ;
\draw[o-] (-1.25,2.5,-2) -- (wall_label) ;
\draw[o-] (-1,1.2,1.5) -- (room_label) node[right] {room};
\draw[o-] (0.35,1.7,2) -- (obstructions) node[above] {obstructions};
\draw[o-] (0.0,3.1,1.7) -- (obstructions) node[above] {obstructions};
\draw[o-] (2.2,0,3) -- (obstructions) ;



\end{tikzpicture}

%% file: Text/Problem_Formulation_and_Optimal_Solution.tex
\section{Problem Formulation and Optimal Solution}
\subsection{Problem Formulation}\label{Subsection: Problem Formulation}
The objective is to deploy the transmit antennas such that the receiver is powered reliably anywhere in the environment.
Analytically, this is attained by determining the optimal transmit antenna placement, which satisfies the constraints imposed on the positions in Section \ref{Subsection: Transmit Antenna Model}, such that the received power $\gamma$ is maximised at the worst possible receiver location. 
Supposing the $i$-th transmit antenna and the receiver are located at positions $(a_i,0,z_0)$ and $(x,y,z)$, respectively, then the received power $\gamma$, defined in \eqref{eq: RX Signal Power}, is proportional to the following function
\begin{align}\label{eq: f_xyz}
    f_{xyz}({a_{1},...,a_{N_t}}) = \sum_{i=1}^{N_t} \frac{1}{(x-a_{i})^2 + y^2 + (z-z_0)^2},
\end{align}
which depends on the transmit antenna positions. 
The proposed design is obtained as the solution of the following max-min problem
\begin{subequations}\label{Problem: Original Problem before Simplifications}
\begin{alignat}{2}
&\underset{{a_{1},...,a_{N_t}}}{\text{maximise}}
&\qquad& \!\min_{x,y,z} f_{xyz}({a_{1},...,a_{N_t}}) 
\label{Objective: Original Problem before Simplifications}\\
&\text{subject to} 
&& \eqref{eq: Receiver position in X},
\eqref{eq: Receiver position in Y},
\eqref{eq: Receiver position in Z},
\eqref{eq: TX antenna positions in X},
\label{con: Set of constraints before Simplifications}
\end{alignat}
\end{subequations}
which is non-convex due to the objective function in \eqref{Objective: Original Problem before Simplifications}. Note that \eqref{eq: f_xyz} is independent of the wavelength, which is required to define the Fresnel region of a wireless system \cite{Zhang2021}. Here, we assume the wavelength of the system is chosen such that the system is operating in the Fresnel region for a given transmit antenna placement. 

By first identifying the receiver locations resulting in the worst performance independent of the transmit antenna positions, the complexity of the problem can be reduced. The set containing these critical receiver locations is denoted by $\mathcal{X^{\mathrm{crit}}} \times \mathcal{Y^{\mathrm{crit}}} \times \mathcal{Z^{\mathrm{crit}}}$, where $\mathcal{X^{\mathrm{crit}}} \subset \mathcal{X}$, $\mathcal{Y^{\mathrm{crit}}} \subset \mathcal{Y}$, and $\mathcal{Z^{\mathrm{crit}}} \subset \mathcal{Z}$.

\begin{proposition}\label{proposition: Restriction to necessary Locations}
It suffices to consider the receiver locations defined by $y=L_y$, $z=-L_z/2$ and $y=L_y$, $z=L_z/2$ for optimisation problem \eqref{Problem: Original Problem before Simplifications}. $\mathcal{Z^{\mathrm{crit}}}$ is given by $\mathcal{Z^{\mathrm{crit}}}=\{-L_z/2\}$ if $0\leq z_0\leq L_z/2$ and $\mathcal{Z^{\mathrm{crit}}}=\{L_z/2 \}$ if $-L_z/2 \leq z_0\leq 0$. $\mathcal{Y^{\mathrm{crit}}}$ is given by $\mathcal{Y^{\mathrm{crit}}}=\{L_y \}$.
\end{proposition}
\begin{proof}
Since \eqref{eq: f_xyz} is a monotonically decreasing function of $y$ and $\vert z-z_0 \vert$, the worst objective value is obtained at the lines at $y=L_y$, $z=-{L_z}/2$ if $0\leq z_0\leq L_z/2$ and $y=L_y$, $z=L_z/2$ if $-{L_z}/2 \leq z_0\leq 0$. Therefore, the critical sets are given by $\mathcal{Y^{\mathrm{crit}}}=\{L_y \}$ and $\mathcal{Z^{\mathrm{crit}}}=\{-L_z/2\}$ if $z_0 \geq 0$ and $\mathcal{Z^{\mathrm{crit}}}=\{L_z/2 \}$ if $z_0 \leq 0$.
\end{proof}

$\mathcal{X^{\mathrm{crit}}}$ cannot be obtained based on monotonicity due to the dependency between the receiver location and the transmit antenna variables in \eqref{eq: f_xyz}. Restricting the receiver positions to the critical set and defining $L_z^\prime = L_z + 2 \vert z_0 \vert$, allows the following equivalent reformulation of \eqref{eq: f_xyz} 
\begin{align}\label{eq: f_x}
    f_{x}({a_{1},...,a_{N_t}}) =  \sum_{i=1}^{N_t} \frac{1}{(x-a_{i})^2 + L_y^2 + \frac{{L_z^\prime}^2}{4}},
\end{align}
which is independent of variables $y$ and $z$. Finally, \eqref{Problem: Original Problem before Simplifications} is equivalently reformulated as
\begin{subequations}\label{Problem: Original Problem after Simplifications}
\begin{alignat}{2}
&\underset{{a_{1},...,a_{N_t}}}{\text{maximise}}
&\qquad& \!\min_{x} f_{x}({a_{1},...,a_{N_t}}) 
\label{Objective: Original Problem after Simplifications}\\
&\text{subject to} 
&& \eqref{eq: Receiver position in X},
\eqref{eq: TX antenna positions in X}.
\label{con: Set of constraints after Simplifications}
\end{alignat}
\end{subequations}

\subsection{Optimal Solution}\label{Subsection: Optimal Solution}
For the remainder of this paper, the number of transmit antennas is set to $N_t=2$. 
This allows for an intuitive illustration of the geometrical considerations underlying the proposed method, thus offering insight into the solution structure of the considered transmit antenna placement problem.
The application of the method presented in this paper to systems with $N_t>2$ is discussed in Subsection \ref{Subsection: More than 2 antennas}.
In the following, the symmetry of $f_{x}$ in \eqref{eq: f_x} is investigated in Proposition \ref{proposition: Condition on optimal positions for N_t = 2}.
\begin{proposition}\label{proposition: Condition on optimal positions for N_t = 2}
    For $N_t = 2$, the optimal locations of the transmit antennas $a_1$ and $a_2$ must satisfy $a_1^* = -a_2^*$.
\end{proposition}
\begin{proof}
For any $a_1,a_2 \in \mathcal{X}$ with $a_2>a_1$, the $\alpha$-superlevel set of $f_x$ is defined as
\begin{align}
    S_\alpha(f_x) = \{ x \in \mathcal{X} \; \vert \; f_x(a_1, a_1 + m) \geq \alpha \},
\end{align}
where $m=a_2-a_1$.
By inspection, the set $S_\alpha(f_x)$ is symmetrical around the plane $x=a_1 +m/2$ for any choice of $\alpha$.
The max-min problem \eqref{Problem: Original Problem after Simplifications} is equivalent to maximising the infimum of $S_\alpha(f_x)$ with respect to (w.r.t.) $\alpha$ while forcing $S_\alpha(f_x)$ to envelope the whole environment $\mathcal{X} \times \mathcal{Y} \times  \mathcal{Z}$. This is equivalent to the following optimisation problem
\begin{subequations}
\begin{alignat}{2}
&\underset{\alpha}{\text{maximise}}
&\qquad& \!\inf S_\alpha(f_{x}) 
\label{Objective: Equivalent Problem}\\
&\text{subject to} 
&& \mathcal{X}\times\mathcal{Y}\times\mathcal{Z} \subseteq S_\alpha(f_x).
\label{con: Envelope Constraint}
\end{alignat}
\end{subequations}
Since the amount of received power decreases as the distance between the receiver and the transmit antennas increases, the level sets departing from the transmit antenna positions must decrease. Since $S_\alpha(f_x)$ must encapsulate the entire room, the maximum value of $\alpha$ is only attained if the symmetry plane passes through the point $(0,0,a_{z_0})$, i.e., $a_1 +m/2=0$, which implies $a_1^* = -a_2^*$.
Intuitively any other choice of $x$, i.e. $0\neq x=a_1 +m/2$, would bias some receiver positions over others, thus causing the performance at the worst position to decrease. 
\end{proof}
Proposition \ref{proposition: Condition on optimal positions for N_t = 2} allows dropping the dependency on $a_2$ in the objective function $f_x$, i.e., $f_x(a_1,a_2)=f_x(a_1)$, with $a_2=-a_1$. For a given $x$, the optimal transmit antenna positions are obtained by determining the stationary points of function $f_x(a_1)$ as follows
\begin{align}\label{eq: Optimality condition for a_1}
    \frac{\partial f_x(a_1)}{\partial a_1} \overset{!}{=} 0.
\end{align}
Solving \eqref{eq: Optimality condition for a_1} for $a_1$, five stationary points $a_1^{\beta}(x,L_y,{L_z^\prime})$, $\beta \in \{ (\text{I}), (\text{II}), (\text{III}), (\text{IV}), (\text{V}) \}$, are found and given in the following
\begin{align}\label{eq: Critical positions for a_1}
    a_1^{(\text{I})}(x,L_y,{L_z^\prime}) &= \frac{1}{2}\sqrt{e(x,L_y,{L_z^\prime}) - d(x,L_y,{L_z^\prime})} \nonumber \\ 
    &= -a_1^{(\text{II})}(x,L_y,{L_z^\prime}), 
\end{align}
\begin{align}\label{eq: Critical positions for a_1_part2}
    a_1^{(\text{III})}(x,L_y,{L_z^\prime}) &= \frac{1}{2}\sqrt{-e(x,L_y,{L_z^\prime}) - d(x,L_y,{L_z^\prime})} \nonumber \\ 
    &= -a_1^{(\text{IV})}(x,L_y,{L_z^\prime}), 
\end{align}
\begin{align}\label{eq: Critical positions for a_1_part3}
    a_1^{(\text{V})}(x,L_y,{L_z^\prime}) &= 0,
\end{align}
where $e(x,L_y,{L_z^\prime}) = 4x\sqrt{4x^2 + 4L_y^2 + {L_z^\prime}^2}$ and $d(x,L_y,{L_z^\prime}) = \left(4x^2+4L_y^2 + {L_z^\prime}^2\right)$. Using the symmetrical relationships $e(-x,L_y,{L_z^\prime}) = -e(x,L_y,{L_z^\prime})$ and $d(-x,L_y,{L_z^\prime}) = d(x,L_y,{L_z^\prime})$, we obtain $a_1^{(\text{III})}(x,L_y,{L_z^\prime}) = a_1^{(\text{I})}(-x,L_y,{L_z^\prime})$ and $a_1^{(\text{IV})}(x,L_y,{L_z^\prime}) = a_1^{(\text{II})}(-x,L_y,{L_z^\prime})$. Thus, it is sufficient to investigate the three stationary points $a_1^{(\text{I})}(x,L_y,{L_z^\prime})$, $a_1^{(\text{II})}(x,L_y,{L_z^\prime})$, and $a_1^{(\text{V})}(x,L_y,{L_z^\prime})$.
Additionally, due to the symmetrical relationship between the antennas, i.e., $a_2=-a_1$, $a_1^{(\text{II})}(x,L_y,{L_z^\prime})$ is also redundant. 

\begin{proposition}\label{Proposition: antenna architecture}
    If the geometry of the environment satisfies $4(L_y/L_x)^2 + ({L_z^\prime}/L_x)^2 \geq 3$, the optimal transmit antenna positions are $a_1^* = a_2^* = 0$, otherwise they satisfy $a_1^* = -a_2^* \neq 0$.
\end{proposition}
\begin{proof}
    The optimal antenna positions must be real numbers. While the stationary point $a_1^{(\text{V})}$ is always real, $a_1^{(\text{I})}(x,L_y,{L_z^\prime})$ is only real and different from $a_1^{(\text{V})}(x,L_y,{L_z^\prime})$, if $e(x,L_y,{L_z^\prime}) > d(x,L_y,{L_z^\prime})$ is satisfied. This holds for 
    \begin{align}\label{eq: condition on geometry}
        \vert x \vert > \sqrt{ \frac{{L_z^\prime}^2}{12} + \frac{L_y^2}{3}} .
    \end{align}
    Consequently, if the inequality in \eqref{eq: condition on geometry} does not hold, then \eqref{eq: Critical positions for a_1} is complex, and thus, unsuitable for describing antenna locations. If the inequality does not hold for the largest possible value of $x$, it does not hold for any value of $x$. Therefore, by inserting $x=L_x/2$ into \eqref{eq: condition on geometry}, a condition purely in terms of the geometry of the system is established:
    \begin{align}\label{eq: condition on geometry 2}
        4 \frac{L_y^2}{L_x^2} + \frac{{L_z^\prime}^2}{L_x^2} \geq 3.
    \end{align}
    If \eqref{eq: condition on geometry 2} is satisfied, then only $a_1^{(\text{V})}(x,L_y)$ is real and the optimal positions for both antennas are $a_1 = a_2 = 0$.
    However, if $4 (L_y/L_x)^2 + ({L_z^\prime}/L_x)^2 < 3$ holds, then $a_1^{(\text{I})}(x,L_y,{L_z^\prime})$ is real and different from zero. 
\end{proof}

In summary, Proposition \ref{Proposition: antenna architecture} shows that a co-located or a distributed antenna architecture is optimal depending on a linear combination of the ratios $L_y / L_x$ and ${L_z^\prime} / L_x$.

\begin{lemma}\label{lemma: Inverse quadratic functions}
    If the geometry of the environment satisfies $4(L_y/L_x)^2 + ({L_z^\prime}/L_x)^2 \geq 3$, then $\mathcal{X^{\mathrm{crit}}} = \{ -L_x/2, L_x/2 \}$.
\end{lemma}
\begin{proof}
Every term in the function $f_x$ is the inverse of a quadratic function \cite{bron}.
Since $f_x$ is a (positive) sum of two positive terms, $f_x$ must also be positive for all inputs, i.e., $\forall x,\forall a_1: f_x(a_1)>0$.

Suppose $a_1^{(\text{V})}$ is the only real stationary point, i.e., \eqref{eq: condition on geometry 2} is satisfied. The global maximum is reached at $a_1^{(\text{V})}$, since only positive values contribute to $f_x$. 
Since there are no other stationary points, the function decays monotonically from the maximum. Therefore, the lowest values occur at the boundaries of $\mathcal{X}$, i.e., $x=-L_x / 2=L_x / 2$. Thus, $\mathcal{X^{\mathrm{crit}}} = \{ -L_x/2, L_x/2 \}$.
\end{proof}

\begin{lemma}\label{lemma: Different critical sets}
    If $4\,(L_y/L_x)^2 + ({L_z^\prime}/L_x)^2 < 3$, the critical points $\mathcal{X^{\mathrm{crit}}}$ are either $\mathcal{X^{\mathrm{crit}}} = \{ -L_x/2, L_x / 2 \}$ or $\mathcal{X^{\mathrm{crit}}} = \{ -L_x/2 , 0, L_x / 2 \}$. The optimal transmit antenna positions are given by 
    \begin{align}\label{eq: Critical values boundaries worse than origin}
        a_1^* &= \frac{L_x}{2} \sqrt{2\, \sqrt{ 4\, \frac{L_y^2}{L_x^2} + \frac{{L_z^\prime}^2}{L_x^2} + 1} - \left(4\, \frac{L_y^2}{L_x^2} + \frac{{L_z^\prime}^2}{L_x^2} + 1 \right)},
    \end{align} 
    and 
    \begin{align}\label{eq: Critical values origin and boundaries equal}
        a_1^* = \frac{L_x}{2} \frac{\sqrt{4 \, \frac{L_y^2}{L_x^2} + \frac{{L_z^\prime}^2}{L_x^2} + 1 }}{\sqrt{3}}, 
    \end{align}
    in the former and the latter case, respectively.
\end{lemma}
\begin{proof}
    If \eqref{eq: condition on geometry 2} is not satisfied, then $f_x$ has real stationary points at $a_1^{(\text{I})}=-a_1^{(\text{II})}$ and $a_1^{(\text{V})}$. 
    As the superposition of inverse quadratic functions, $f_x$ is asymptotic to zero when the function's argument goes to infinity \cite{bron}. Consequently, $a_1^{(\text{I})}=-a_1^{(\text{II})}$ must be maxima.
    Since the value of the function decays monotonically from any maximum, and $a_1^{(\text{V})}$ is a stationary point which lies exactly between the two maxima, it must be a local minimum. This is due to the symmetry of $f_x$ around the centre of $\mathcal{X}$. 
    In this case, two scenarios are possible. Either $x=-L_x / 2=L_x / 2$ remains to be the worst possible location (as in Lemma \ref{lemma: Inverse quadratic functions}) or the objective value at an additional, different $x$ location becomes critical.
    In the first case, the critical set becomes $\mathcal{X^{\mathrm{crit}}} = \{ -L_x/2 , L_x / 2 \}$. The optimal transmit antenna position for $a_1^*$ is obtained by inserting $x=L_x / 2$ into \eqref{eq: Critical positions for a_1} which yields \eqref{eq: Critical values boundaries worse than origin}.
    The latter case occurs when the transmit antennas are moved so far towards the boundaries of $\mathcal{X}$ that the performance is favoured there, causing the performance at the origin to decrease the most. Consequently, the worst performance is achieved either at the extremes or at the centre of the interval $\mathcal{X}$. In this case, the set of critical values $\mathcal{X^{\mathrm{crit}}}$ becomes $\mathcal{X^{\mathrm{crit}}} = \{ -L_x/2 , 0, L_x / 2 \}$.
    By definition, the performance at the critical points $x \in \mathcal{X^{\mathrm{crit}}}$ is worst. Therefore, the performance at all critical points must be identical.
    Since $f_{-\frac{L_x}{2}}(a_1) = f_{\frac{L_x}{2}}(a_1)$ holds for any $a_1$ by symmetry, it is sufficient to investigate the equality
    \begin{align}
        f_0(a_1^*) \overset{!}{=} f_{\frac{L_x}{2}}(a_1^*) 
    \end{align}
    which implies \eqref{eq: Critical values origin and boundaries equal}.    
\end{proof}

From Lemma \ref{lemma: Different critical sets} it follows that $a_1^*$ depends purely on the geometry of the system. 
%
Next, in Lemma \ref{lemma: point of transition} a condition on the geometry is established which describes the point of transition from \eqref{eq: Critical values boundaries worse than origin} to \eqref{eq: Critical values origin and boundaries equal}.
\begin{lemma}\label{lemma: point of transition}
    The point of transition between \eqref{eq: Critical values boundaries worse than origin} and \eqref{eq: Critical values origin and boundaries equal} occurs at $4 (L_y/L_x)^2 + ({L_z^\prime}/L_x)^2 = 5/4$.
\end{lemma}
\begin{proof}
    The point of transition is identified by equating \eqref{eq: Critical values boundaries worse than origin} and \eqref{eq: Critical values origin and boundaries equal}.
\end{proof}

\begin{proposition}\label{proposition: Final result}
    The optimal transmit antenna position $a_1^*$ such that the received power $\gamma$ is maximised at the worst receiver position in the environment is given by
    \begin{align}\label{eq: Optimal summary}
    a_1^* =
    \begin{cases}
      \eqref{eq: Critical values origin and boundaries equal} &\text{if $4 \frac{L_y^2}{L_x^2} + \frac{{L_z^\prime}^2}{L_x^2} \leq \frac{5}{4}$}, \\
      \eqref{eq: Critical values boundaries worse than origin} &\text{if $\frac{5}{4} \leq 4 \frac{L_y^2}{L_x^2} + \frac{{L_z^\prime}^2}{L_x^2} \leq 3$}, \\
      0 &\text{if $4 \frac{L_y^2}{L_x^2} + \frac{{L_z^\prime}^2}{L_x^2} \geq 3$}.
    \end{cases}
\end{align}
\end{proposition}
\begin{proof}
    The proof follows from Lemma \ref{lemma: Different critical sets}, Lemma \ref{lemma: point of transition}, and \eqref{eq: Critical positions for a_1}.
\end{proof}

\begin{corollary}
    The optimal transmit antenna deployment, defined in Proposition \ref{proposition: Final result}, yields the maximum received power $\gamma^*$ for the worst receiver position which is given by
    \begin{align}\label{eq: Optimal received power}
    \gamma^* = Pc
    \begin{cases}
        \frac{2}{\frac{4}{3}  L_y^2 + \frac{1}{12}  L_x^2 + \frac{1}{3}{L_z^\prime}^2} &\text{if $4 \frac{L_y^2}{L_x^2} + \frac{{L_z^\prime}^2}{L_x^2} \leq \frac{5}{4}$}, \\
        \frac{2\sqrt{4\frac{L_y^2}{L_x^2} + \frac{{L_z^\prime}^2}{L_x^2}+1}+2}{4L_y^2 + {L_z^\prime}^2} &\text{if $\frac{5}{4} \leq 4 \frac{L_y^2}{L_x^2} + \frac{{L_z^\prime}^2}{L_x^2} \leq 3$}, \\
        \frac{2}{L_y^2 + \frac{1}{4}  L_x^2 + \frac{1}{4}  {L_z^\prime}^2 } &\text{if $4 \frac{L_y^2}{L_x^2} + \frac{{L_z^\prime}^2}{L_x^2} \geq 3$}.
    \end{cases}
    \end{align}
\end{corollary}
\begin{proof}
    The received power $\gamma^*$ is obtained by inserting \eqref{eq: f_xyz} into \eqref{eq: RX Signal Power} for the optimal transmit antenna positions \eqref{eq: Optimal summary}
    \begin{align}
        \gamma^*  = P c f_x(a_1^*,a_2^*),
    \end{align}
    where $f_x(a_1^*,a_2^*)$ is the optimal objective value of the optimisation problem in \eqref{Problem: Original Problem after Simplifications}. The objective value is obtained by inserting \eqref{eq: Optimal summary} into \eqref{eq: f_x} and evaluating the resulting function at $x \in \mathcal{X^{\mathrm{crit}}}$. Then, $f_x(a_1^*,a_2^*)$ is given by
    \begin{align}\label{eq: max rx power 3}
        \frac{2}{\frac{4}{3}  L_y^2 + \frac{1}{12}  L_x^2 + \frac{1}{3}{L_z^\prime}^2} \;\; \text{if} \;\; 4\frac{L_y^2}{L_x^2} + \frac{{L_z^\prime}^2}{L_x^2} \leq \frac{5}{4},
    \end{align}
    \begin{align}\label{eq: max rx power 2}
        \frac{2\sqrt{4\frac{L_y^2}{L_x^2} + \frac{{L_z^\prime}^2}{L_x^2}+1}+2}{4L_y^2 + {L_z^\prime}^2} \;\; \text{if} \;\; \frac{5}{4} \leq 4\frac{L_y^2}{L_x^2} + \frac{{L_z^\prime}^2}{L_x^2} \leq 3,
    \end{align}
    and
    \begin{align}\label{eq: max rx power 1}
     \frac{2}{L_y^2 + \frac{1}{4}  L_x^2 + \frac{1}{4}  {L_z^\prime}^2 } \;\; \text{if} \;\; 4 \frac{L_y^2}{L_x^2} + \frac{{L_z^\prime}^2}{L_x^2} \geq 3.
    \end{align}
\end{proof}

\subsection{Power Gain over Optimal Far-Field Solution}
Next, the possible power gain of employing the proposed optimal antenna positioning over the optimal far-field position is quantified. The ideal position for all transmit antennas under the far-field assumption lies in the centre of $\mathcal{X}$. This follows from approximating the path losses corresponding to different transmit antennas with a constant value which is analogous to the co-located transmit antenna architecture where the optimal transmit antenna positions are $a_1^*=a_2^*=0$.
The performance gain $\eta$ which is obtained by designing the system using the exact SWM instead of the approximated PWM depends on the geometry of the environment. Therefore, the gain $\eta$ is computed by comparing the performance of using the optimal transmit antenna locations in the Fresnel region \eqref{eq: Optimal summary} to the optimal ones for the far-field, which are obtained by placing both antennas at zero. 

\begin{corollary}
The gain that is obtained by using the optimal transmit antenna placement depends on the geometry of the environment and is given by
\begin{align}\label{eq: Gain summary}
    \eta =
    \begin{cases}
        \frac{12\frac{L_y^2}{L_x^2} +3\frac{{L_z^\prime}^2}{L_x^2}+3} {16 \frac{L_y^2}{L_x^2}+4\frac{{L_z^\prime}^2}{L_x^2}+1} &\text{if $4\frac{L_y^2}{L_x^2} + \frac{{L_z^\prime}^2}{L_x^2} \leq \frac{5}{4}$}, \\
        \frac{\left(4\frac{L_y^2}{L_x^2}+\frac{{L_z^\prime}^2}{L_x^2}+1\right)^{\frac{3}{2}} + 4\frac{L_y^2}{L_x^2}+\frac{{L_z^\prime}^2}{L_x^2}+1}{16\frac{L_y^2}{L_x^2}+4\frac{{L_z^\prime}^2}{L_x^2}} &\text{if $\frac{5}{4} \leq 4\frac{L_y^2}{L_x^2} + \frac{{L_z^\prime}^2}{L_x^2} \leq 3$}, \\
        1 &\text{if $4\frac{L_y^2}{L_x^2} + \frac{{L_z^\prime}^2}{L_x^2} \geq 3$}.
    \end{cases}
\end{align}
\end{corollary}
\begin{proof}
The optimal objective value is given by \eqref{eq: max rx power 3}-\eqref{eq: max rx power 1} and depends on the geometry of the environment.
In the far-field, the worst receiver position lies at the boundaries of $\mathcal{X}$ since the antennas are placed at $x=0$. Thus, the objective value of the far-field solution corresponds to \eqref{eq: max rx power 1}.
\end{proof}

\subsection{Systems with $N_t>2$}\label{Subsection: More than 2 antennas}
The methodology outlined in this paper may be extended to systems with a larger number of transmit antennas, i.e., $N_t > 2$. By virtue of the max-min problem, an asymmetrical transmit antenna placement is suboptimal. Therefore, for an uneven number of transmit antennas $N_t$, one optimal transmit antenna position must lie in the centre of $\mathcal{X}$. Moreover, when $N_t$ is even, Proposition \ref{proposition: Condition on optimal positions for N_t = 2} must hold for pairs of transmit antennas to ensure a symmetrical transmit antenna placement.
The approach for identifying all critical receiver positions in $\mathcal{X}^\mathrm{crit}$ may be cumbersome for large $N_t$, as the number of critical positions $\vert \mathcal{X}^\mathrm{crit} \vert$ is expected to grow with $N_t$. Consequently, the number of transition points, which are necessary for formulating the optimal positions, is expected to grow. 
Alternatively, the optimal transmit antenna placement may then be obtained by solving \eqref{Problem: Original Problem after Simplifications} numerically. A numerical method of solving the problem is discussed at the end of Section \ref{Subsection: Result antennas}.

%% file: Text/Results_and_Performance_Evalutation.tex
\section{Results and Performance Evaluation}
\subsection{Optimal Position of $a_1^*$}\label{Subsection: Result antennas}
The optimal position of $a_1^*$ \eqref{eq: Optimal summary} is visualised as a function of the geometry of the system in Fig. \ref{fig: Optimal a_1 position vs. L_y/L_x}. 
\begin{figure}[t]
    \centering
    \scalebox{0.82}{
    \input{Figures/solution_L_y_over_L_x.tex}
    }
    \caption{Optimal placement of antenna $a_1^*$ \eqref{eq: Optimal summary}.}
    \label{fig: Optimal a_1 position vs. L_y/L_x}
\end{figure}
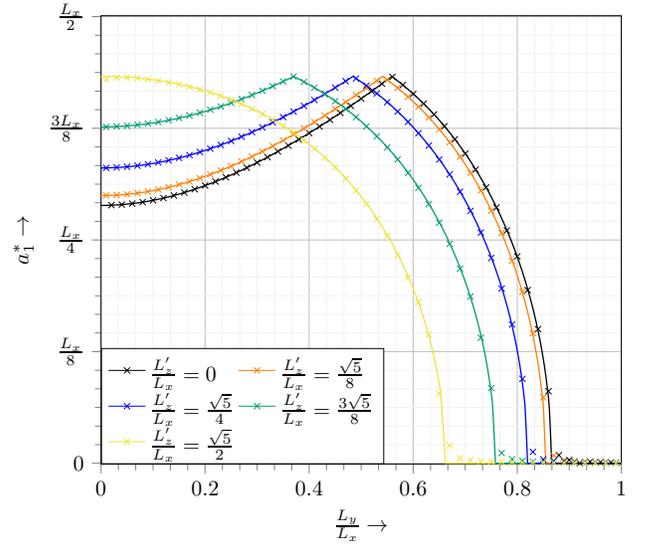
To this end, the optimal position of $a_1^*$ is depicted as a function of one variable and one parameter. The optimal position $a_1^*$ is given as a function of $L_y / L_x$, i.e., $a_1^*(L_y / L_x)$, while the value of ${L_z^\prime} / L_x$ is fixed and the parameter value is indicated by the legend entry corresponding to the respective plot. The function $a_1^*(L_y / L_x)$ is plotted by considering the parameter values ${L_z^\prime} / L_x = \{ 0, \sqrt{5}/8, \sqrt{5}/4, 3\sqrt{5}/8, \sqrt{5}/2 \}$. The parameter values are chosen to represent the boundaries of the possible geometrical ratios of the environment. For ${L_z^\prime} / L_x = 0$, the environment collapses into a two-dimensional room. The maximum value ${L_z^\prime} / L_x = \sqrt{5}/2$ follows from the condition in Lemma \ref{lemma: point of transition}. 

Fig. \ref{fig: Optimal a_1 position vs. L_y/L_x} provides insight into how the geometrical properties of the environment impact the optimal transmit antenna positions. By basing the proposed method on the SWM, an optimal solution was obtained which is capable of capturing the effects in the Fresnel region and naturally converges to the far-field solution as $4 L_y^2 / L_x^2 + {L_z^\prime}^2 / L_x^2$ grows. 
In the far-field, the impact of the spherical nature of the EM wavefronts at the worst receiver locations diminishes and the PWM becomes sufficiently accurate. 
Consequently, when the optimal solution for the Fresnel region coincides with the far-field solution, i.e., $4 L_y^2 / L_x^2 + {L_z^\prime}^2 / L_x^2 \geq 3$, then the WPT system does not have to be designed for the Fresnel region.
However, for $4 L_y^2 / L_x^2 + {L_z^\prime}^2 / L_x^2 \leq 3$, designing the system for the Fresnel region has to be ensured, which requires considering the relationship between the wavelength and the distance between the transmit antennas in comparison to their distance from the receiver. 
The analytical solution was validated by solving the problem numerically. To this end, the quadratic form method in \cite{Shen2018} for Fractional Programming was extended to max-min-sum-ratio problems and a sequence of parameterised convex optimisation problems was solved using CVXPY \cite{diamond2016cvxpy}. The numerical results are indicated by the crosses overlaying the respective analytical results in Fig. \ref{fig: Optimal a_1 position vs. L_y/L_x}. Fig. \ref{fig: Optimal a_1 position vs. L_y/L_x} also shows that the numerical method suffers from numerical inaccuracies as the optimal $a_1^*$ approaches $0$.

\subsection{Power Gain over Far-Field Solution}\label{Subsection: Result gain}
The power gain $\eta$ in \eqref{eq: Gain summary} is visualised in Fig. \ref{fig: Gain} as a function of $L_y/L_x$, while fixing the parameter ${L_z^\prime}/L_x$ for every plot of the function. The values of the parameters are identical to the ones listed in Subsection \ref{Subsection: Result antennas}.
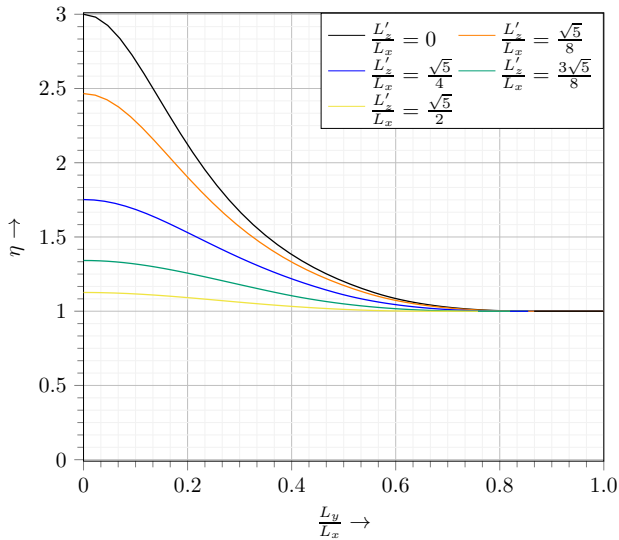
\begin{figure}[t]
    \centering
    \scalebox{0.82}{
    \input{Figures/gain.tex}
    }
    \caption{Possible power gain $\eta$ \eqref{eq: Gain summary}.}
    \label{fig: Gain}
\end{figure}
The maximum gain is three and is achieved for $L_y / L_x \rightarrow 0$ with ${L_z^\prime} / L_x = 0$. Consequently, using the proposed optimal positioning based on the SWM provides up to three times the amount of received power compared to the far-field approach. As the impact of the spherical nature of the EM wavefronts decreases, the gain over the far-field solution drops.

%% file: Figures/solution_L_y_over_L_x.tex
\definecolor{black}{RGB}{0,0,0}
\definecolor{bluish-green}{RGB}{0,158,115}
\definecolor{yellow}{RGB}{240, 228, 66}

\begin{tikzpicture}

\begin{axis}[
scale only axis,
tick align=outside,
tick pos=left,
unbounded coords=jump,
xlabel={$\frac{L_y}{L_x} \rightarrow$},
xmin=0.0, xmax=1,
ylabel={$a_1^* \rightarrow$},
ymin=-0.001, ymax=0.501,
ytick={0,0.125,0.25,0.375,0.5},
yticklabels={$0$,$\frac{L_x}{8}$,$\frac{L_x}{4}$,$\frac{3L_x}{8}$,$\frac{L_x}{2}$},
grid=both,
grid style={line width=.1pt, draw=gray!10},
major grid style={line width=.2pt,draw=gray!50},
minor tick num=5,
legend cell align=left,
legend pos=south west,
legend style={
    at={(0,0)},
    legend columns=2,
    legend entries={$\frac{L_z^\prime}{L_x} = 0$, $\frac{L_z^\prime}{L_x} =\frac{\sqrt{5}}{8}$, $\frac{L_z^\prime}{L_x} =\frac{\sqrt{5}}{4}$, $\frac{L_z^\prime}{L_x} =\frac{3\sqrt{5}}{8}$, $\frac{L_z^\prime}{L_x} =\frac{\sqrt{5}}{2}$}
    }
]
\coordinate (centre) at (0.68,0.25);
\addlegendimage{mark=x,black};
\addlegendimage{mark=x,orange};
\addlegendimage{mark=x,blue};
\addlegendimage{mark=x,bluish-green};
\addlegendimage{mark=x,yellow};



\addplot [
line cap=round,
smooth,
semithick,
black
] file[skip first]{Figures/antenna_data/L_y_over_L_x/param0.dat}
coordinate[above, pos=0.65] (0);

\addplot [
black,
mark=x,
only marks
] file[skip first]{Figures/antenna_data/simulation_data/0.dat};

\addplot [
line cap=round,
smooth,
semithick,
orange
] file[skip first]{Figures/antenna_data/L_y_over_L_x/param1.dat}
coordinate[above, pos=0.6] (1);

\addplot [
orange,
mark=x,
only marks
] file[skip first]{Figures/antenna_data/simulation_data/1.dat};

\addplot [
line cap=round,
smooth,
semithick,
blue
] file[skip first]{Figures/antenna_data/L_y_over_L_x/param2.dat}
coordinate[above, pos=0.53] (2);

\addplot [
blue,
mark=x,
only marks
] file[skip first]{Figures/antenna_data/simulation_data/2.dat};

\addplot [
line cap=round,
smooth,
semithick,
bluish-green,
] file[skip first]{Figures/antenna_data/L_y_over_L_x/param3.dat}
coordinate[above, pos=0.65] (3);

\addplot [
bluish-green,
mark=x,
only marks
] file[skip first]{Figures/antenna_data/simulation_data/3.dat};

\addplot [
line cap=round,
smooth,
semithick,
yellow
] file[skip first]{Figures/antenna_data/L_y_over_L_x/param4.dat}
coordinate[left, pos=0.57] (4);

\addplot [
yellow,
mark=x,
only marks
] file[skip first]{Figures/antenna_data/simulation_data/4.dat};

\addplot [
yellow, 
semithick,
domain=0.661860081130729:1
]{(0)};

\addplot [
bluish-green, 
semithick,
domain=0.757979588402178:1
]{(0)};

\addplot [
blue, 
semithick,
domain=0.819765022276868:1
]{(0)};

\addplot [
orange, 
semithick,
domain=0.854695075999124:1
]{(0)};

\addplot [
black, 
semithick,
domain=0.866025403784439:1
]{(0)};


\end{axis}

\end{tikzpicture}

%% file: Figures/gain.tex
\definecolor{black}{RGB}{0,0,0}
\definecolor{bluish-green}{RGB}{0,158,115}
\definecolor{yellow}{RGB}{240, 228, 66}

\begin{tikzpicture}

\begin{axis}[
scale only axis,
tick align=outside,
tick pos=left,
unbounded coords=jump,
xlabel={$\frac{L_y}{L_x} \rightarrow$},
xmin=0.0, xmax=2,
xtick={0,0.4,0.8,1.2,1.6,2.0},
xticklabels={$0$,$0.2$,$0.4$,$0.6$,$0.8$,$1.0$},
ylabel={$\eta \rightarrow$},
ymin=-0.01, ymax=3.01,
grid=both,
grid style={line width=.1pt, draw=gray!10},
major grid style={line width=.2pt,draw=gray!50},
minor tick num=5,
legend cell align=left,
legend style={
    at={(1,1)},
    legend columns=2,
    legend entries={$\frac{L_z^\prime}{L_x} = 0$, $\frac{L_z^\prime}{L_x} =\frac{\sqrt{5}}{8}$, $\frac{L_z^\prime}{L_x} =\frac{\sqrt{5}}{4}$, $\frac{L_z^\prime}{L_x} =\frac{3\sqrt{5}}{8}$, $\frac{L_z^\prime}{L_x} =\frac{\sqrt{5}}{2}$}
    }
]
\coordinate (centre) at (0.68,0.25);
\addlegendimage{black};
\addlegendimage{orange};
\addlegendimage{blue};
\addlegendimage{bluish-green};
\addlegendimage{yellow};

\coordinate (centre2) at (1,1.1);

\addplot [
black,
semithick,
] file[skip first]
{Figures/gain_data/L_z_over_L_x/param0.dat}
coordinate[left, pos=0.64] (0);
\label{eta}

\addplot [
orange,
semithick,
] file[skip first]{Figures/gain_data/L_z_over_L_x/param1.dat}
coordinate[above, pos=0.4] (1);

\addplot [
blue,
semithick,
] file[skip first]{Figures/gain_data/L_z_over_L_x/param2.dat}
coordinate[above, pos=0.4] (2);

\addplot [
bluish-green,
semithick,
] file[skip first]{Figures/gain_data/L_z_over_L_x/param3.dat}
coordinate[above, pos=0.23] (3);

\addplot [
yellow,
semithick,
] file[skip first]
{Figures/gain_data/L_z_over_L_x/param4.dat}
coordinate[above, pos=0.08] (4);

\addplot [
yellow, 
semithick,
domain=1.32456337559024:2
]{(1)};

\addplot [
bluish-green, 
semithick,
domain=1.51637341261138:2
]{(1)};

\addplot [
blue, 
semithick,
domain=1.63970028785408:2
]{(1)};

\addplot [
orange, 
semithick,
domain=1.70943097506076:2
]{(1)};

\addplot [
black, 
semithick,
domain=1.73205080756888:2
]{(1)};



\end{axis}

\end{tikzpicture}

%% file: Text/Conclusion.tex
\section{Conclusion}\label{Section: Conclusion}
In this paper, we considered a MISO WPT system that comprises a two-antenna energy transmitter and a single antenna energy receiver under LoS conditions. Hereby, the objective was the maximisation of the received power for the worst possible receiver position by identifying the optimal transmit antenna deployment. 
The approach in this paper leverages the symmetry among the positions of the transmit antennas. The symmetry is a necessary constraint in order to avoid biasing certain receiver positions over others, as this would lead to a worse objective globally. 
The proposed optimal, analytical solution provides insight into how the geometry impacts the optimal placement of the transmit antennas.
Our solution reveals that distributing the transmit antennas in the environment is optimal when in any location of the environment the PWM is not a suitable approximation of the SWM. Otherwise, a co-located antenna architecture is optimal which corresponds to the optimal solution for the far-field. 
The solution was validated by solving the problem numerically. Furthermore, the maximum power gain over the far-field solution was found to be three.
The extensions of the proposed solution to small-scale fading channels, multiple transmit antennas, and more flexible transmit antenna placements are interesting topics for future work.